\pdfoutput=1 

\documentclass[sigconf, nonacm]{acmart}

\usepackage[nolist,nohyperlinks]{acronym}
\usepackage{algorithm}
\usepackage[capitalise]{cleveref}
\usepackage{algpseudocode}
\usepackage[binary-units=true]{siunitx}
\usepackage{amsthm}

\newcommand*{\symDefine}[2]{\newcommand{{#1}}{{#2}}}

\symDefine{\symServiceA}{A}
\symDefine{\symServiceB}{B}
\symDefine{\symBase}{\alpha}
\symDefine{\symDownSamplingFunction}{d}
\symDefine{\symInterval}{I}
\symDefine{\symIndexI}{i}
\symDefine{\symIndexJ}{j}
\symDefine{\symIndexK}{k}
\symDefine{\symNumSampleRates}{m}
\symDefine{\symNumSampleRatesInSample}{n}
\symDefine{\symQuantity}{q}
\symDefine{\symQuantityEstimator}{Q}
\symDefine{\symNaiveQuantityEstimator}{\symQuantityEstimator_{\mathrm{naive}}}
\symDefine{\symNewQuantityEstimator}{\symQuantityEstimator_{\mathrm{new}}}
\symDefine{\symRandom}{R}
\symDefine{\symSampleRateFunction}{r}
\symDefine{\symRateLimit}{f_\mathrm{max}}
\symDefine{\symSpan}{s}
\symDefine{\symSpanA}{\symSpan_\mathrm{a}}
\symDefine{\symSpanB}{\symSpan_\mathrm{b}}
\symDefine{\symDeltaTime}{\Delta T}
\symDefine{\symSampleRate}{u}
\symDefine{\symSampleRateInSample}{v}
\symDefine{\symSumQuantity}{w}
\symDefine{\symSumQuantityEstimate}{W}
\symDefine{\symSetOfSpans}{\mathcal{X}}
\symDefine{\symZ}{z}
\symDefine{\symSetOfSampledSpans}{\mathcal{Y}}
\symDefine{\symOtherSet}{\mathcal{A}}

\symDefine{\symDesiredSampleRate}{\rho}

\DeclareMathOperator*{\symVariance}{Var}
\DeclareMathOperator*{\symExpectation}{\mathbb{E}}
\DeclareMathOperator*{\symBigO}{\mathcal{O}}
\DeclareMathOperator*{\symProbability}{Pr}

\allowdisplaybreaks
\begin{document}
\title{Estimation from Partially Sampled Distributed Traces} 

\author{Otmar Ertl}
\affiliation{%
  \institution{Dynatrace Research}
  \city{Linz}
  \state{Austria}
}
\email{otmar.ertl@dynatrace.com}

\begin{abstract}
Sampling is often a necessary evil to reduce the processing and storage costs of distributed tracing. In this work, we describe a scalable and adaptive sampling approach that can preserve events of interest better than the widely used head-based sampling approach. Sampling rates can be chosen individually and independently for every span, allowing to take span attributes and local resource constraints into account. The resulting traces are often only partially and not completely sampled which complicates statistical analysis. To exploit the given information, an unbiased estimation algorithm is presented. Even though it does not need to know whether the traces are complete, it reduces the estimation error in many cases compared to considering only complete traces.
\end{abstract}

\maketitle

\pagestyle{plain}

\section{Introduction}
Distributed tracing \cite{Shkuro2019, Parker2020, Kaldor2017, Sigelman2010, Greifeneder2012} has become an important tool for monitoring and failure diagnosis of large distributed systems. 
It measures the behavior and the performance of individual operations also called spans. By propagating parent span contexts including a unique trace identifier, that is generated at the endpoint (the entry point into the monitored system), to child spans allows to derive the hierarchical causal relationships between collected spans. 
In this way, the set of collected spans sharing the same trace identifier will form a directed acyclic graph, which is called a trace. These traces allow to see, find, and understand far-reaching dependencies across process boundaries and are invaluable when investigating failures or performance issues.

Tracing involves some inherent performance overhead and leads to additional processing and storage costs. 
Span contexts must be sent from parent to child spans across process boundaries to enable linking spans to traces. Measuring latencies or response times and capturing exceptions or errors result in higher resource consumption. Transmitting observed data over the network to a preprocessor or a storage layer further increases costs. If monitoring and capturing everything is not an option, on-demand collection or sampling are the remaining alternatives.
The first approach does not allow for post-mortem analysis. Repetitions of the same failure must occur after activating trace collection to gain further insight. Demand-driven collection corresponds to sampling with rates 0 or 1, depending on whether the collection is off or on, respectively.
In the general case, the free choice of the sampling rate allows a smooth adjustment of the collection behavior between those two extremes. The larger the sampling rate, the more information is collected, and the higher the chance that problems and failures can be analyzed and explained by the recorded data. 

To minimize the loss of information, a clever sampling strategy is crucial. For example, stratified sampling \cite{Cochran1977} can be used to collect interesting events like slow service calls or errors at higher rates than frequent events representing normal behavior to reduce the variance. A major complication is that at the time of the sampling decision, it is often not known whether a trace is of interest or not. 

\subsection{Overview}
After elaborating the state of the art like head-based sampling and tail-based sampling, we will present a different approach called partial trace sampling. Unlike head-based sampling, which uses the same sampling decision for all spans of the trace, partial trace sampling generalizes this by allowing the sampling rate to be set individually and independently for every span. This flexibility enables better adaptation to local resource constraints and recording meaningful or infrequently accessed spans with higher probabilities.

Based on a mathematical model describing the new sampling approach, a general estimation algorithm is developed that allows the unbiased estimation of arbitrary quantities from partially sampled traces, like how often some service $\symServiceA$ called some other service $\symServiceB$. The results of this work allow deriving requirements and recommendations for monitoring frameworks like OpenTelemetry \cite{OpenTelemetry} to produce sampled trace data that is more useful for downstream statistical analysis.

As we have found that the performance of the proposed estimation algorithm strongly depends on the choice of sampling rates, a restriction to a discrete set of values makes sense. In particular, limiting sampling rates to powers of $\frac{1}{2}$ also has other benefits we will discuss. Unfortunately, this restriction conflicts with traditional rate-limiting sampling \cite{Shkuro2019}. Therefore, we present a workaround that can adhere to a prescribed discrete set of sampling rates.
Furthermore, as partial trace sampling can lead to fragmented traces, we briefly discuss which information could be added to the span context to reduce the estimation error for queries relying on hierarchical relationships.

\subsection{Head-Based Trace Sampling}
Head-based sampling \cite{Shkuro2019,Sambasivan2016, LasCasas2018, LasCasas2019}, also called up-front sampling \cite{Parker2020}, is a simple and widely used sampling approach for traces. The sampling decision for the whole trace is made at the endpoint and propagated together with the unique trace identifier as part of the span context to descendant spans which respect the sampling decision made for the root span. The simplicity of head-based sampling allows straightforward estimation. The reciprocal of the sampling rate used for the root span corresponds to the extrapolation factor, also called adjusted weight \cite{Cohen2009}, that applies to the whole trace.

Choosing an appropriate sampling rate is challenging. At the endpoint, which is the origin of a trace, not much is known. The attributes of the root span often do not reveal which branches will be called. Furthermore, they rarely indicate anomalous backend behavior and are therefore often not suited for filtering interesting traces. Furthermore, the sampling rate must be chosen small enough not to violate any constraints on data collection at backend services. Without knowledge about what is going on in the backend the choice of sampling rates may be difficult. Therefore, a feedback loop is often required to allow dynamic adaptation of the sampling behavior \cite{Shkuro2019}. This requires an orchestration service with a back channel to control the sampling rates at endpoints.
This rather indirect way of controlling the amount of collected spans leads to a slow adaptation to new situations. In particular, sudden jumps in load can violate the limits of data collection, and sudden drops in load can lead to more loss of information than necessary. In addition, head-based sampling is not capable of responding quickly to back pressure caused by the span processing pipeline.

\subsection{Tail-Based Trace Sampling}
Tail-based sampling \cite{Shkuro2019,Sambasivan2016, LasCasas2018, LasCasas2019}, also called response-based sampling  \cite{Parker2020} is another popular technique to reduce the number of traces that need to be stored. In contrast to head-based sampling, this approach collects all spans, and only after reconstructing the trace, a sampling decision is made. Since all information of a trace is available, the sampling rate can be chosen more intelligently than with head-based sampling. However, this comes with significant additional network and memory costs, because all spans of the same trace have to be sent to the same instance of a preprocessing service that buffers them until the trace has completed. Additional complexity is introduced if the preprocessing service needs to be scalable. Adding or removing instances requires rerouting spans of a trace to balance the work load which carries a risk of obtaining incomplete traces. Tail-based sampling can be combined with other more upstream sampling techniques like head-based sampling or partial trace sampling which is described in the following.

\subsection{Partial Trace Sampling}
If network costs are a major concern or rerouting and buffering spans is not an option, the sampling decision must be made immediately for each individual observed span. Independent sampling decisions for spans of the same trace would lead to a very small probability to record all spans of a trace simultaneously. Therefore, the sampling decision of different spans has to be consistent to increase the probability to collect all spans of a trace. Head-based sampling, for example, achieves consistency by propagating the sampling decision to child spans. 

Another way to enforce consistent sampling is to share a uniform random number $\symRandom\in[0,1)$ across all spans of a trace and sample the span only if the sampling condition $\symRandom<\symSampleRateFunction$ is satisfied \cite{Sigelman2010}. If the same sampling rate $\symSampleRateFunction$ is used for all spans, the result will be equivalent to that of head-based sampling with rate $\symSampleRateFunction$. 
Generating a random number at the root and propagating it as part of the span context together with the trace identifier to child spans  is one option to realize a common random number \cite{Molkova2020}. Alternatively, to avoid the propagation of an additional field, a hash value calculated from the trace identifier can be used instead \cite{Sigelman2010}. For that, a fast high-quality hash function \cite{Urban} is needed with outputs that can be regarded as uniform random values. If the trace identifiers are already generated using some high-quality random number generator, they could be directly mapped to the range $[0,1)$ and used as random numbers. Similar to propagating a random number this would save computation costs as hashing the trace identifier for every single span is avoided. 

If the sampling rate is not chosen homogeneously across all spans, the sampling decision is no longer the same for all spans, and as a result, traces may be only partially sampled. Google's tracing system called Dapper \cite{Sigelman2010}, for example, allows increasing the sampling rates for child spans \cite{MacDonald2021}. In this way, rarely called branches are sampled with a higher probability and more details are preserved that would otherwise have been lost with pure head-based sampling. This can lead to incomplete traces since only subgraphs of the same trace might be captured.

If backend services are called by many different endpoints or many times within the same trace, it is also useful to allow reducing the sampling rates for child spans. Hence, in the most general case, we want the sampling rates to be chosen independently for every single span as exemplified in \cite{Molkova2020}. This, of course, will lead to a fragmentation of traces, if spans in the middle are not recorded. Nevertheless, there is still a reasonable chance to record complete traces, if the spans are sampled consistently using a shared random number as described before. By definition, the smallest sampling rate used by any span corresponds to the probability that a trace is fully recorded. 

Unlike head-based sampling, partial trace sampling does not need global orchestration to set appropriate sampling rates at endpoints. They can be chosen independently, for example, to satisfy local constraints. While local constraints also limit the maximum sampling rate of head-based sampling, partial trace sampling is more flexible and allows using larger sampling rates for other spans and thus collecting more information. The number of complete traces would be the same for both, head-based and partial sampling. However, the additional fragments of other traces collected can reduce the estimation error. Often it is not necessary to know the full trace. Sometimes problems can be explained by looking at certain branches only. Even if different fragments are not connected because some span in the middle has not been sampled, we know that they must belong to the same trace because of the common trace identifier. Therefore, partially sampled traces are still useful to explain far-reaching dependencies.

To make the most out of the given information, incomplete traces need to be incorporated for estimation, which is not trivial in the general case. In particular, correct estimation of quantities that depend on multiple spans of a trace is not obvious. For example, estimating how often some service $\symServiceA$ called another service $\symServiceB$ also depends on the sampling rate of the intermediate spans connecting service $\symServiceA$ and $\symServiceB$, because only when collecting the hierarchical connection between $\symServiceA$ and $\symServiceB$, we really know that $\symServiceA$ called $\symServiceB$. To the best of our knowledge, there is no prior work that considers estimation from partially and consistently sampled traces. Therefore, we will present an estimation algorithm capable of providing unbiased estimates for arbitrary quantities of interest.

\section{Methodology}

Partial trace sampling allows to choose an individual sampling rate for every single span of a trace. The sampling rate function $\symSampleRateFunction:\symSetOfSpans\rightarrow (0,1]$ will describe the corresponding sampling rate for every span $\symSpan\in\symSetOfSpans$. Here, $\symSetOfSpans$ denotes the set of all spans of a given trace. The sampling rate of span $\symSpan\in\symSetOfSpans$ is given by $\symSampleRateFunction(\symSpan)$. If $\symRandom\in[0,1)$ is the random number shared across all spans of the trace and used for consistent sampling decisions, the span $\symSpan$ will be sampled if and only if 
\begin{equation}
\label{equ:sampling_decision}
\symRandom < \symSampleRateFunction(\symSpan) \Leftrightarrow \text{span}\ \symSpan\ \text{is sampled}.
\end{equation}
The sampling rates $\symSampleRateFunction(\symSpan)\in(0,1]$ can be freely chosen for each trace. However, we do not allow zero sampling rates, as unbiased estimation of something that is never sampled is impossible. Furthermore, we require that the sampling rates do not depend on the shared random number $\symRandom$. This implies, that it is not allowed to choose the sampling rate based on sampling decisions of other spans of the same trace. If $\symRandom$ is calculated from the trace identifier, the chosen sampling rate must also not depend on the trace identifier. 

However, the sampling rate may depend on any other span attributes. This allows, for example, to sample spans with long durations or with errors at a higher frequency. The sampling rate may even depend on attributes or sampling rates of other spans if the corresponding information is propagated and hence available when the sampling rate is needed for the sampling decision. The calculation of the sampling rate may also incorporate any local or global resource constraints.

For a trace with set of spans $\symSetOfSpans$, partial trace sampling can be mathematically described using a sampling function $\symDownSamplingFunction$ that maps $\symSetOfSpans$ to a corresponding subset of sampled spans 
\begin{equation*}
\symSetOfSampledSpans=\symDownSamplingFunction(\symSetOfSpans;\symRandom,\symSampleRateFunction)\subseteq\symSetOfSpans.
\end{equation*}
The sampling function $\symDownSamplingFunction$ is defined as
\begin{equation*}
\symDownSamplingFunction(\symSetOfSpans;\symRandom,\symSampleRateFunction)
:=
\lbrace\symSpan\in\symSetOfSpans: \symRandom<\symSampleRateFunction(\symSpan)\rbrace 
\end{equation*}
and also requires the shared uniform random number $\symRandom\in[0,1)$ and the sampling rate function $\symSampleRateFunction$.

Obviously, in the special case that $\symSampleRateFunction(\symSpan)$ is constant for all spans $\symSpan$ of a trace, either all or none of the spans will be sampled. This corresponds to the behavior of head-based sampling. To have the same sampling rate for every span of a trace, the sampling rate chosen for the root span must be propagated to descendent spans. This is in contrast to traditional head-based sampling where usually just a flag as specified for the W3C trace context \cite{W3CTraceContext} indicating the sampling decision of the root is propagated. Passing on the sampling rate has the advantage of immediately knowing the sampling rate and thus also the adjusted weight for each span. If only a flag representing the sampling decision is propagated, the sampling rate for non-root spans will not be known until the corresponding root span is available, from which the sampling rate can eventually be retrieved.

Partial trace sampling will result in a fully sampled trace, if the random number $\symRandom$ is smaller than all individual span sampling rates, mathematically expressed as
\begin{equation*}
\symDownSamplingFunction(\symSetOfSpans; \symRandom, \symSampleRateFunction) = \symSetOfSpans
\Leftrightarrow
\symRandom<\min_{\symSpan\in\symSetOfSpans} \symSampleRateFunction(\symSpan).
\end{equation*}
The corresponding probability is given by the minimum sampling rate
\begin{equation}
\label{equ:complete_sampled_probability}
\symProbability(\symDownSamplingFunction(\symSetOfSpans; \symRandom, \symSampleRateFunction) = \symSetOfSpans)= \min_{\symSpan\in\symSetOfSpans} \symSampleRateFunction(\symSpan).
\end{equation}

Unfortunately, for a given set of sampled spans, it is often not obvious whether a trace has been completely sampled or not. Furthermore, the sampling rates of discarded spans are usually not reported and hence the minimum sampling rate might not be known. Especially, if only some leaf spans are missing, the situation will not be clear. This is different from head-based sampling where the sampling rate is guaranteed to be equal for all spans and traces are always fully sampled. If it is important to know whether a trace is complete, partial trace sampling would require collecting extra information like the number of children of each span. Luckily, knowing the completeness of traces is often not necessary as shown in the following sections.

\subsection{Trace Analysis}
\label{sec:trace_analysis}
Traces are collected for various purposes. The analysis of individual traces, such as those associated with a particular failed action on a web page, is one important use case. When sampling is active, regardless of the method used, there is no guarantee that the traces of interest were recorded. If partial trace sampling is used, there is at least a higher chance that at least parts of the trace have been sampled. Traces can be easily sorted by completeness \cite{Molkova2020}. By definition, the traces with the smallest shared random number are the most complete.

Collective trace analysis is another important use case. Instead of looking at individual traces, information is extracted from all available traces using statistical means. Typically, we would like to know the sum of a certain quantity $\symQuantity$ over all traces represented as sets of spans $\symSetOfSpans_\symIndexI$
\begin{equation*}
\symSumQuantity(\symSetOfSpans_1,\symSetOfSpans_2,\ldots) = \sum_{\symIndexI}\symQuantity(\symSetOfSpans_\symIndexI).
\end{equation*}
Here $\symQuantity=\symQuantity(\symSetOfSpans)$ is a function that maps the set of spans $\symSetOfSpans$ of a given trace to a real number. Depending on what we would like to know, $\symQuantity$ can be arbitrarily defined, as demonstrated by the following examples:
\begin{itemize}
  \item If we are interested in the total count of traces we would simply have $\symQuantity(\symSetOfSpans):=1$.
  \item If we want to count traces with a certain property, $\symQuantity$ would be the corresponding indicator function, returning 1 or 0 dependent on whether the trace given by the set of spans $\symSetOfSpans$ has that property or not, respectively. For example, we might be interested in the number of traces calling some particular service $\symServiceA$. If $\symSetOfSpans$ contains at least one span corresponding to a call to $\symServiceA$, the returned value would be 1.
  \item If we are interested in the total number of spans, we would choose $\symQuantity(\symSetOfSpans):=|\symSetOfSpans|$, where $|\symSetOfSpans|$ denotes the cardinality of the set of spans $\symSetOfSpans$.
  \item If we want to know the number of spans that have a certain property, we would define $\symQuantity(\symSetOfSpans):=|\symSetOfSpans\cap\symOtherSet|$, where $\symOtherSet$ denotes the universe of spans with the desired property. For example, one might be interested in the number of spans which are associated with some specific error.
  \item If we want to determine the average call tree depth, $\symQuantity(\symSetOfSpans)$ would be defined to return the depth of the graph spanned by the set of spans $\symSetOfSpans$. Together with the estimate for the total count of traces an estimate of the average depth can be obtained.
\end{itemize}

\subsection{Estimation Problem}
When traces are sampled partially, only subsets of spans $\symSetOfSampledSpans_\symIndexI=\symDownSamplingFunction(\symSetOfSpans_\symIndexI;\symSampleRateFunction_\symIndexI,\symRandom_\symIndexI)\subseteq\symSetOfSpans_\symIndexI$ will be available for estimation. Complicating matters, we are often unable to verify their completeness if, for example, only leaf spans are missing. Moreover, the existence of traces, for which not a single span has been collected, is not even known.
Therefore, it is challenging to find an estimator $\symQuantityEstimator = \symQuantityEstimator(\symSetOfSampledSpans_\symIndexI)$, which provides an unbiased estimate for $\symQuantity(\symSetOfSpans_\symIndexI)$ based on the partially sampled set of spans $\symSetOfSampledSpans_\symIndexI$ and thus satisfies
$\symQuantity(\symSetOfSpans_\symIndexI)=\symExpectation\nolimits_{\symRandom_\symIndexI}(\symQuantityEstimator(\symSetOfSampledSpans_\symIndexI))$. 
We must have $\symQuantityEstimator(\emptyset)=0$, since a trace with none of its spans sampled cannot contribute to the estimate. 

If we have found such an unbiased estimator,
the composite estimator
\begin{equation*}
\symSumQuantityEstimate(\symSetOfSampledSpans_1, \symSetOfSampledSpans_2,\ldots) = \sum_{\symIndexI:\symSetOfSampledSpans_\symIndexI\neq\emptyset}\symQuantityEstimator(\symSetOfSampledSpans_\symIndexI)
\end{equation*}
will consequently be an unbiased estimator for $\symSumQuantity$, hence $\symExpectation\nolimits_{\symRandom_1,\symRandom_2,\ldots}(\symSumQuantityEstimate(\symSetOfSampledSpans_1, \symSetOfSampledSpans_2,\ldots)) = \symSumQuantity(\symSetOfSpans_1,\symSetOfSpans_2,\ldots)$. 
If the random values $\symRandom_\symIndexI$, that are used for sampling different traces, are independent, the corresponding variance will be given by
\begin{equation*}
\symVariance\nolimits_{\symRandom_1,\symRandom_2,\ldots}(\symSumQuantityEstimate(\symSetOfSampledSpans_1, \symSetOfSampledSpans_2,\ldots))
=
\sum_{\symIndexI}\symVariance\nolimits_{\symRandom_\symIndexI}(\symQuantityEstimator(\symSetOfSampledSpans_\symIndexI)).
\end{equation*}
As a consequence, the estimation problem boils down to finding an unbiased estimator $\symQuantityEstimator = \symQuantityEstimator(\symSetOfSampledSpans_\symIndexI)$ for the quantity of interest $\symQuantity(\symSetOfSpans_\symIndexI)$ from a single trace, which allows to omit the subscript $\symIndexI$ in the following.

\subsection{Naive Estimation Approach}
A simple approach is to use only complete traces for estimation. This requires knowing whether the set of sampled spans $\symSetOfSampledSpans$ is equal to the true set of spans $\symSetOfSpans$ of a given trace, which may often not be the case, as mentioned earlier. The naive estimator weights the quantity $\symQuantity(\symSetOfSampledSpans)$ of complete traces by the inverse of the probability that $\symSetOfSampledSpans=\symSetOfSpans$ given by $\symProbability(\symSetOfSampledSpans=\symSetOfSpans)=\frac{1}{\min_{\symSpan\in\symSetOfSampledSpans} \symSampleRateFunction(\symSpan)}$ according to  \eqref{equ:complete_sampled_probability}
\begin{equation}
\label{equ:naive_estimator}
\symNaiveQuantityEstimator(\symSetOfSampledSpans)
=
\begin{cases}
\symQuantity(\symSetOfSampledSpans)\frac{1}{\min_{\symSpan\in\symSetOfSampledSpans} \symSampleRateFunction(\symSpan)} & \text{if}\ \symSetOfSampledSpans = \symSetOfSpans, \\
0 & \text{if}\ \symSetOfSampledSpans \subset \symSetOfSpans.
\end{cases}
\end{equation}
By definition, this estimator is unbiased $\symExpectation\nolimits_\symRandom(\symNaiveQuantityEstimator(\symSetOfSampledSpans))=\symQuantity(\symSetOfSpans)$ and it can be shown that its variance is
\begin{equation*}
\symVariance\nolimits_\symRandom(\symNaiveQuantityEstimator(\symSetOfSampledSpans)) = \symQuantity(\symSetOfSpans)^2 \left(\frac{1}{\min_{\symSpan\in\symSetOfSpans} \symSampleRateFunction(\symSpan)}-1\right).
\end{equation*}

The naive estimator does not incorporate partially sampled traces and hence does not use all information available. Therefore, we will present another unbiased estimation approach that does not require knowing the completeness of traces and performs provably better when the quantity of interest satisfies certain conditions.

\subsection{Quantity Function}
For the simple estimator, the quantity function $\symQuantity$ needs to be defined only for sets of spans $\symSetOfSpans$ that represent complete traces. In the following we require $\symQuantity$ to be also defined for all nonempty subsets of spans $\symSetOfSampledSpans$ with $\emptyset\subset \symSetOfSampledSpans \subseteq \symSetOfSpans$. As long as $\symSetOfSampledSpans$ is distinguishable from the set of spans of a complete trace, we allow the value of $\symQuantity(\symSetOfSampledSpans)$ to be arbitrarily chosen. Otherwise, $\symQuantity(\symSetOfSampledSpans)$ must be equal to the true value of that trace.

As discussed later, it is beneficial if the extended function $\symQuantity$ satisfies certain conditions. Therefore, we call $\symQuantity$ \emph{bounded}, if $\emptyset\subset \symSetOfSampledSpans \subseteq \symSetOfSpans$ with $\symSetOfSpans$ denoting the set of spans of a complete trace implies
\begin{equation}
\label{equ:bounded}
\symQuantity(\symSetOfSampledSpans) \in [0, 2\symQuantity(\symSetOfSpans)].
\end{equation}
for all $\symSetOfSampledSpans$. In case $\symQuantity(\symSetOfSpans)<0$, the interval $[0, 2\symQuantity(\symSetOfSpans)]$ is interpreted as $[2\symQuantity(\symSetOfSpans),0]$. Furthermore, we call $\symQuantity$ \emph{monotonic}, if $\emptyset\subset \symSetOfSampledSpans_1 \subseteq \symSetOfSampledSpans_2\subseteq\symSetOfSpans$ implies
\begin{equation}
\label{equ:monotonic}
\symQuantity(\symSetOfSampledSpans_1) \in [0, \symQuantity(\symSetOfSampledSpans_2)].
\end{equation}
for all $\symSetOfSampledSpans_1$ and $\symSetOfSampledSpans_2$. Obviously, if $\symQuantity$ is monotonic, it will also be bounded, which can easily seen when setting $\symSetOfSampledSpans_1=\symSetOfSampledSpans$ and $\symSetOfSampledSpans_2=\symSetOfSpans$. 

Many quantities are inherently monotonic, such as the number of spans of a trace $\symQuantity(\symSetOfSpans)=|\symSetOfSpans|$.
However, there are quantities, which are not even bounded. 
For example, consider a trace that invokes two different services $\symServiceA$ and $\symServiceB$ in that order. The set of spans of this trace is $\symSetOfSpans = \lbrace\symSpanA, \symSpanB\rbrace$. Suppose we are interested in counting all traces that called $\symServiceA$ but not $\symServiceB$. Obviously, we would then have $\symQuantity(\symSetOfSpans)=0$. For the subset $\symSetOfSpans' = \lbrace \symSpanA\rbrace$ we must have $\symQuantity(\symSetOfSpans')=1$, if there is no additional information indicating whether $\symSetOfSpans'$ represents a complete trace or not. In the case of doubt, it must be assumed that the trace is complete. It follows that $0= \symQuantity(\symSetOfSpans)<\symQuantity(\symSetOfSpans')=1$, while $\emptyset\subset\symSetOfSpans'\subset \symSetOfSpans$ showing that $\symQuantity$ is not bounded and thus not monotonic as well.

In the case where it is always known whether a set of spans represents a completely sampled trace, $\symQuantity$ can always be defined to be monotonic by simply returning zero for incomplete sets of spans, hence $\symSetOfSampledSpans \subset \symSetOfSpans\Rightarrow \symQuantity(\symSetOfSampledSpans)=0$.
This is very similar to the naive estimator $\symNaiveQuantityEstimator$ \eqref{equ:naive_estimator}, which ignores incomplete traces. In fact, the new estimator we will propose is equivalent to $\symNaiveQuantityEstimator$, if the definition of $\symQuantity$ is extended this way.

\subsection{Estimation by Example}
To illustrate the basic idea of the new estimator, we first discuss a concrete example. Suppose we are interested in estimating the total number of spans, $\symQuantity(\symSetOfSpans)=|\symSetOfSpans|$. Furthermore, consider a trace $\symSetOfSpans=\lbrace\symSpan_\text{p},\symSpan_\text{c}\rbrace$ with exactly two spans, a parent (root) span $\symSpan_\text{p}$ and a child span $\symSpan_\text{c}$. Both are sampled with rates 
$\symSampleRateFunction(\symSpan_\text{p})$ and $\symSampleRateFunction(\symSpan_\text{c})$, respectively, satisfying $0 < \symSampleRateFunction(\symSpan_\text{c}) < \symSampleRateFunction(\symSpan_\text{p}) < 1$. Dependent on the shared random number $\symRandom$, three different outcomes are possible for the sampled trace $\symSetOfSampledSpans=\symDownSamplingFunction(\symSetOfSpans;\symRandom,\symSampleRateFunction)\subseteq\symSetOfSpans$:
\begin{enumerate}
\item $\symRandom\in[\symSampleRateFunction(\symSpan_\text{p}),1)\Rightarrow \symSetOfSampledSpans=\emptyset$: Since none of both spans is sampled, there is nothing to contribute to the estimate and we have $\symQuantityEstimator(\symSetOfSampledSpans)=0$. As $\symRandom$ is uniformly distributed, this case occurs with a probability of $\symProbability(\symSetOfSampledSpans = \emptyset) = 1-\symSampleRateFunction(\symSpan_\text{p})$.
\item $\symRandom\in[\symSampleRateFunction(\symSpan_\text{c}), \symSampleRateFunction(\symSpan_\text{p}))\Rightarrow \symSetOfSampledSpans=\lbrace\symSpan_\text{p}\rbrace$: Only the parent span is sampled in this case, which occurs with a probability of $\symProbability(\symSetOfSampledSpans=\lbrace\symSpan_\text{p}\rbrace) = \symSampleRateFunction(\symSpan_\text{p})-\symSampleRateFunction(\symSpan_\text{c})$. Since there is no information about child spans, we also do not know $\symSampleRateFunction(\symSpan_\text{c})$. Using just the information available, namely the parent span including its sampling rate $\symSampleRateFunction(\symSpan_\text{p})$, we have to assume that the trace originally consisted of only a single span. Hence, we have to weight this span with $1/\symSampleRateFunction(\symSpan_\text{p})$ to account for the unseen spans in order to get a supposed unbiased estimate given by $\symQuantityEstimator(\symSetOfSampledSpans)=\symQuantityEstimator(\lbrace\symSpan_\text{p}\rbrace) = 1/\symSampleRateFunction(\symSpan_\text{p})$.
\item $\symRandom\in[0, \symSampleRateFunction(\symSpan_\text{c}))\Rightarrow \symSetOfSampledSpans=\lbrace\symSpan_\text{p},\symSpan_\text{c}\rbrace=\symSetOfSpans$: Both spans are sampled. This occurs with a probability of $\symProbability(\symSetOfSampledSpans=\lbrace\symSpan_\text{p},\symSpan_\text{c}\rbrace)=\symSampleRateFunction(\symSpan_\text{c})$. If we would proceed as in the previous case, weighting the two observed spans with $1/\symSampleRateFunction(\symSpan_\text{c})$ and using the estimate $\symQuantityEstimator(\lbrace\symSpan_\text{p},\symSpan_\text{c}\rbrace) = 2/\symSampleRateFunction(\symSpan_\text{c})$, we would introduce a bias, because when calculating the weighted average over all three cases, we would get $\symExpectation_\symRandom(\symQuantityEstimator(\symSetOfSampledSpans)) = \symProbability(\symSetOfSampledSpans = \emptyset)\cdot\symQuantityEstimator(\emptyset)+\symProbability(\symSetOfSampledSpans=\lbrace\symSpan_\text{p}\rbrace)\cdot\symQuantityEstimator(\lbrace\symSpan_\text{p}\rbrace)+\symProbability(\symSetOfSampledSpans=\lbrace\symSpan_\text{p},\symSpan_\text{c}\rbrace)\cdot\symQuantityEstimator(\lbrace\symSpan_\text{p},\symSpan_\text{c}\rbrace) = 3 -\symSampleRateFunction(\symSpan_\text{c})/\symSampleRateFunction(\symSpan_\text{p})\neq 2$, which is not what we would expect for an unbiased estimator. We must set $\symQuantityEstimator(\lbrace\symSpan_\text{p},\symSpan_\text{c}\rbrace) = 1/\symSampleRateFunction(\symSpan_\text{c}) + 1/\symSampleRateFunction(\symSpan_\text{p})$ to fix the bias introduced by the previous case and to get $\symExpectation_\symRandom(\symQuantityEstimator(\symSetOfSampledSpans))=2$.
\end{enumerate}
This simple example demonstrates the basic idea of our new estimation approach. More complete traces must account for the estimation errors that are made with less complete traces.

\subsection{New Estimation Approach}
In the general case, we consider the set of distinct sampling rates of all spans in $\symSetOfSpans$ denoted in ascending order by $\symSampleRate_1 < \symSampleRate_2 < \ldots < \symSampleRate_{\symNumSampleRates}$, where $\symNumSampleRates$ is the total number of those values.
Dependent on the random number $\symRandom$, the spans $\symSetOfSpans$ may be sampled in $\symNumSampleRates$ different ways such that the sampled set of spans $\symSetOfSampledSpans$ is nonempty. 
If $\symRandom\in[\symSampleRate_{\symIndexI},\symSampleRate_{\symIndexI+1})$ (with $\symSampleRate_0:=0$ and $0\leq \symIndexI < \symNumSampleRates$) the sampled set of spans will be $\symSetOfSampledSpans = \symDownSamplingFunction(\symSetOfSpans; \symRandom, \symSampleRateFunction)=\symDownSamplingFunction(\symSetOfSpans; \symSampleRate_{\symIndexI}, \symSampleRateFunction)$. 
If $\symRandom\in[\symSampleRate_0,\symSampleRate_1)=[0, \min_{\symSpan\in\symSetOfSpans} \symSampleRateFunction(\symSpan))$ implying $\symSetOfSampledSpans=\symDownSamplingFunction(\symSetOfSpans; 0, \symSampleRateFunction)=\symSetOfSpans$, the trace will be fully sampled.
The different outcomes are distributed according to the probabilities $\symProbability(\symSetOfSampledSpans = \symDownSamplingFunction(\symSetOfSpans; \symSampleRate_{\symIndexI}, \symSampleRateFunction)) = \symSampleRate_{\symIndexI+1}-\symSampleRate_{\symIndexI}$. With a probability of $1 - \symSampleRate_{\symNumSampleRates}$ nothing is sampled, hence $\symProbability(\symSetOfSampledSpans = \emptyset)=1 - \symSampleRate_{\symNumSampleRates}=1-\max_{\symSpan\in\symSetOfSpans} \symSampleRateFunction(\symSpan)$.

Let us focus on the case $\symRandom\in[\symSampleRate_{\symIndexI},\symSampleRate_{\symIndexI+1})$ with $0\leq \symIndexI < \symNumSampleRates$. According to the way consistent sampling works, 
the sampled set of spans $\symSetOfSampledSpans$ will have $\symNumSampleRatesInSample = \symNumSampleRates - \symIndexI$ distinct sampling rates denoted in ascending order 
by $\symSampleRateInSample_1 < \symSampleRateInSample_2 < \ldots < \symSampleRateInSample_{\symNumSampleRatesInSample}$ and given by the relationship $\symSampleRateInSample_\symIndexJ = \symSampleRate_{\symIndexJ + \symIndexI}$. These sampling rates can also be obtained directly from $\symSetOfSampledSpans$ as $\lbrace\symSampleRateInSample_1, \symSampleRateInSample_2, \ldots, \symSampleRateInSample_\symNumSampleRatesInSample\rbrace = \bigcup_{\symSpan\in\symSetOfSampledSpans}\symSampleRateFunction(\symSpan)$. 

The new estimator as function of a nonempty set of consistently sampled spans $\symSetOfSampledSpans$ can then be written as
\begin{align}
&\symNewQuantityEstimator(\symSetOfSampledSpans)
:=
\frac{\symQuantity(\symSetOfSampledSpans)}{\symSampleRateInSample_{1}} - 
\sum_{\symIndexJ=1}^{\symNumSampleRatesInSample-1} 
\symQuantity(\symDownSamplingFunction(\symSetOfSampledSpans;\symSampleRateInSample_{\symIndexJ},\symSampleRateFunction))
\left(\frac{1}{\symSampleRateInSample_{\symIndexJ}}-\frac{1}{\symSampleRateInSample_{\symIndexJ+1}}\right)
\nonumber
\\
&\quad=
\frac{\symQuantity(\symSetOfSampledSpans)}{\symSampleRateInSample_{\symNumSampleRatesInSample}} 
+
\sum_{\symIndexJ=1}^{\symNumSampleRatesInSample-1}
\left(\symQuantity(\symSetOfSampledSpans)-\symQuantity(\symDownSamplingFunction(\symSetOfSampledSpans;\symSampleRateInSample_{\symIndexJ},\symSampleRateFunction))\right)
\left(
\frac{1}{\symSampleRateInSample_{\symIndexJ}}
-
\frac{1}{\symSampleRateInSample_{\symIndexJ+1}}
\right)
\nonumber
\\
&\quad=
\frac{\symQuantity(\symDownSamplingFunction(\symSetOfSampledSpans;\symSampleRateInSample_{\symNumSampleRatesInSample-1},\symSampleRateFunction))
}{\symSampleRateInSample_{\symNumSampleRatesInSample}}
+
\sum_{\symIndexJ=1}^{\symNumSampleRatesInSample-1}
\frac{\symQuantity(\symDownSamplingFunction(\symSetOfSampledSpans;\symSampleRateInSample_{\symIndexJ-1},\symSampleRateFunction))
-
\symQuantity(\symDownSamplingFunction(\symSetOfSampledSpans;\symSampleRateInSample_{\symIndexJ},\symSampleRateFunction))
}{\symSampleRateInSample_{\symIndexJ}}.
\label{equ:new_estimator}
\end{align}
It is easy to show that all these expressions are equivalent. The third expression requires $\symSampleRateInSample_0$ to be defined such that $\symDownSamplingFunction(\symSetOfSampledSpans;\symSampleRateInSample_0,\symSampleRateFunction)=\symSetOfSampledSpans$. Hence, any value from $[0, \symSampleRateInSample_1)$ could be used for $\symSampleRateInSample_0$.

According to \cref{the:unbiased} and \cref{the:variance}, both given and proved in the appendix, this estimator is unbiased 
\begin{equation*}
\symExpectation\nolimits_\symRandom(\symNewQuantityEstimator(\symSetOfSampledSpans)) = \symQuantity(\symSetOfSpans)
\end{equation*}
and has variance
\begin{multline}
\label{equ:variance_new}
\symVariance\nolimits_\symRandom(\symNewQuantityEstimator(\symSetOfSampledSpans))
=
\symQuantity(\symSetOfSpans)^2
\left(
\frac{1}{\symSampleRate_{\symNumSampleRates}}
-
1
\right)
\\
+
\sum_{\symIndexJ=1}^{\symNumSampleRates-1}
(\symQuantity(\symSetOfSpans)-\symQuantity(\symDownSamplingFunction(\symSetOfSpans;\symSampleRate_{\symIndexJ},\symSampleRateFunction)))^2
\left(
\frac{1}{\symSampleRate_{\symIndexJ}}
-
\frac{1}{\symSampleRate_{\symIndexJ+1}}
\right).
\end{multline}
In contrast to the naive estimator $\symNaiveQuantityEstimator$ \eqref{equ:naive_estimator}, $\symNewQuantityEstimator$ does not need to know whether $\symSetOfSampledSpans$ corresponds to a fully sampled trace. If $\symQuantity$ is bounded as defined by \eqref{equ:bounded}, $\symNewQuantityEstimator$ also leads to a smaller estimation error as $\symVariance_\symRandom (\symNewQuantityEstimator(\symSetOfSampledSpans)) \leq \symVariance_\symRandom (\symNaiveQuantityEstimator(\symSetOfSampledSpans))$ according to \cref{lem:comparison}.

If $\symQuantity$ is monotonic as defined by \eqref{equ:monotonic} and a trace is sampled using different sampling rate functions $\symSampleRateFunction_1$ and $\symSampleRateFunction_2$, where the second dominates the first meaning that $\symSampleRateFunction_1(\symSpan)\leq \symSampleRateFunction_2(\symSpan)$ for all $\symSpan\in\symSetOfSpans$, the variance will satisfy $\symVariance\nolimits_\symRandom(\symNewQuantityEstimator(\symDownSamplingFunction(\symSetOfSpans; \symRandom, \symSampleRateFunction_1))) \geq \symVariance\nolimits_\symRandom(\symNewQuantityEstimator(\symDownSamplingFunction(\symSetOfSpans; \symRandom, \symSampleRateFunction_2)))$ according to \cref{the:variance_comparison}. This is not surprising, as higher sample rates lead to more collected spans and more information that can be used to obtain more accurate estimates. 

\subsection{Estimating Counts}
Counting spans or traces that match certain filter criteria is a very common task. For these special cases, the new estimator \eqref{equ:new_estimator} can be significantly simplified. When estimating the number of matching spans from a partially sampled trace, the quantity of interest would be $\symQuantity(\symSetOfSpans) = |\symSetOfSpans \cap \symOtherSet|$, where $\symOtherSet$ denotes the universe of matching spans. It can be shown that estimator \eqref{equ:new_estimator} simplifies to (see \cref{cor:span_extrapolation})
\begin{equation*}
\symNewQuantityEstimator(\symSetOfSampledSpans)=
\sum_{\symSpan\in\symSetOfSampledSpans\cap\symOtherSet}
\frac{1}{\symSampleRateFunction(\symSpan)}.
\end{equation*}
This is actually what we would have expected as it is simply the sum of inverse sampling rates of all matching spans.

When counting matching traces, $\symQuantity$ would correspond to an indicator function. If $\symQuantity$ is monotonic the estimator can be simplified to (see \cref{cor:indicator_estimation})
\begin{equation}
\label{equ:trace_count_extrapolation}
\symNewQuantityEstimator(\symSetOfSampledSpans)
=
\begin{cases}
\frac{1}{\symSampleRateInSample_\symIndexK} & \text{if}\ \symQuantity(\symSetOfSampledSpans) = 1\\
0 & \text{if}\ \symQuantity(\symSetOfSampledSpans) = 0
\end{cases}
\end{equation}
with  $\symIndexK:=\max\lbrace \symIndexJ\in\lbrace 1,2,\ldots,\symNumSampleRatesInSample\rbrace : \symQuantity(\symDownSamplingFunction(\symSetOfSampledSpans;\symSampleRateInSample_{\symIndexJ-1},\symSampleRateFunction))=1 \rbrace$ and $\symSampleRateInSample_0:=0$.
Hence, the extrapolation factor depends on the minimum sampling rate of all spans needed to match the given filter criteria.

\subsection{Practical Considerations}

\cref{alg:estimation} summarizes the new estimation approach based on \eqref{equ:new_estimator}. The number of loop iterations is given by the number of unique sampling rates found in the sampled set of spans $\symSetOfSampledSpans$. In the case where each sampling rate of each span is distinct and the evaluation of $\symQuantity(\symSetOfSampledSpans)$ takes $\symBigO(|\symSetOfSampledSpans|)$ time, the algorithm will exhibit an $\symBigO(|\symSetOfSampledSpans|^2)$ time complexity. Some quantities allow optimizations that reduce the worst case time complexity. However, in the general case, the time complexity can only be reduced, if the number of different sampling rates is limited.

\begin{algorithm}
\caption{Estimation algorithm.}
\label{alg:estimation}
\begin{algorithmic}
\Require $\symSetOfSampledSpans$ \Comment{nonempty set of sampled spans of the same trace}
\State\hspace{0.72cm}$\symSampleRateFunction:\symSetOfSampledSpans\rightarrow(0,1]$\Comment returns sampling rate of given span
\State\hspace{0.72cm}$\symQuantity$\Comment extracts quantity of interest from given set of spans
\Ensure $\symQuantityEstimator$\Comment unbiased estimate for quantity of interest
\State $\symQuantityEstimator \gets 0$
\State $\symQuantity_\text{prev}  \gets \symQuantity(\symSetOfSampledSpans)$
\Loop
\State $\symSampleRateInSample\gets \min_{\symSpan\in\symSetOfSampledSpans} \symSampleRateFunction(\symSpan)$
\State $\symSetOfSampledSpans\gets \lbrace\symSpan \in \symSetOfSampledSpans:\symSampleRateInSample < \symSampleRateFunction(\symSpan)\rbrace $\Comment downsampling step
\State \algorithmicif\ $\symSetOfSampledSpans$ is empty \algorithmicthen\ \Return $\symQuantityEstimator + \symQuantity_\text{prev} / \symSampleRateInSample$
\State $\symQuantity_\text{next}  \gets \symQuantity(\symSetOfSampledSpans)$
\State $\symQuantityEstimator \gets \symQuantityEstimator + (\symQuantity_\text{prev} - \symQuantity_\text{next}) / \symSampleRateInSample$
\State $\symQuantity_\text{prev}\gets \symQuantity_\text{next}$
\EndLoop
\end{algorithmic}
\end{algorithm}

As a solution, we propose a restriction to values of a geometric sequence $\symSampleRateFunction(\symSpan)\in \lbrace 1, \symBase, \symBase^2, \symBase^3,\ldots\rbrace$ with $\symBase\in(0,1)$. 
In this way, the expected number of different sampling rates present in the sample of spans will be constant. 
Assume that the maximum sampling rate of all spans in $\symSetOfSpans$ is $\symBase^{\symIndexK}=\max_{\symSpan\in\symSetOfSpans}\symSampleRateFunction(\symSpan)$ with some $\symIndexK\geq 0$. The sampled set of spans $\symSetOfSampledSpans$ can be nonempty only if $\symRandom<\symBase^{\symIndexK}$ according to \eqref{equ:sampling_decision}. Given $\symRandom<\symBase^{\symIndexK}$, the conditional probability that spans with sampling rate $\symBase^{\symIndexK+\symIndexI}$ are sampled, is $\symProbability(\symRandom < \symBase^{\symIndexK+\symIndexI}\mid \symRandom<\symBase^{\symIndexK})=\symBase^\symIndexI$, where $\symIndexI\geq 0$. Therefore, the expected number of different sampling rates of spans in $\symSetOfSampledSpans$ must be bounded by $\sum_{\symIndexI=0}^\infty \symBase^\symIndexI = \frac{1}{1-\symBase}=\symBigO(1)$.

A further advantage of using a discrete set of sampling rates is that the index $\symIndexJ\geq 0$ can be used to encode the sampling rate $\symBase^\symIndexJ$ instead of its numerical representation. Likewise, the shared random number $\symRandom$ can be encoded as integer as well, because for the sampling decision \eqref{equ:sampling_decision} it is only relevant to which of the intervals $[\symBase^{\symIndexI+1},\symBase^{\symIndexI})$ with $\symIndexI\geq 0$ the shared random value $\symRandom$ belongs to. Therefore, it is sufficient to propagate the index of the corresponding interval instead of the numerical representation of $\symRandom$. The sampling decision simplifies to a plain integer comparison
\begin{equation*}
\symIndexI \geq \symIndexJ \Leftrightarrow \text{span}\ \symSpan\ \text{is sampled},
\end{equation*}
where $\symIndexJ$ denotes the exponent of the sampling rate $\symSampleRateFunction(\symSpan)=\symBase^\symIndexJ$ and $\symIndexI$ indicates $\symRandom\in [\symBase^{\symIndexI+1}, \symBase^\symIndexI)$.
Since $\symRandom$ is uniformly distributed, the index $\symIndexI$ will be geometrically distributed with success probability $(1-\symBase)$ as
\begin{equation*}
\symProbability(\symRandom\in [\symBase^{\symIndexI+1}, \symBase^\symIndexI)) = \symBase^\symIndexI - \symBase^{\symIndexI+1} = \symBase^{\symIndexI}(1-\symBase).
\end{equation*}
It can therefore be drawn directly from the corresponding geometric distribution without having to generate $\symRandom$ first.

According to our experience, $\symBase=\frac{1}{2}$ is a good choice as it gives enough freedom to select the sampling rate while also offering some other nice benefits: 
\begin{itemize}
\item A single byte is able to encode all relevant sampling rates. Similarly, a single byte is sufficient to encode the index of the interval that surrounds the shared random number which reduces the information that needs to propagated to child spans.
\item As already discussed, it is favorable to directly generate $\symIndexI$ indicating the interval $[\symBase^{\symIndexI+1},\symBase^{\symIndexI})$ enclosing the uniform random number $\symRandom$. Drawing a random value from a geometric distribution with a success probability of $\symBase=\frac{1}{2}$ can be efficiently implemented by simply taking the number of leading zeros of a uniformly distributed random integer. This is a very cheap and natively supported operation on most processors.
\item If all possible sampling rates are reciprocals of integers as is the case for all values of a geometric sequence with coefficient $\symBase=\frac{1}{2}$, estimates of integer quantities like counts will be integers as well. $\symQuantity\in\mathbb{Z}$ and $\frac{1}{\symSampleRateInSample_\symIndexJ}\in\mathbb{Z}$ imply $\symQuantityEstimator(\symSetOfSampledSpans)\in\mathbb{Z}$ according to \eqref{equ:naive_estimator} and \eqref{equ:new_estimator}. 
It is desirable to have estimates from the same number system, as consumers of estimated quantities do not have to deal with rounding.
\end{itemize}
The obvious downside of being limited to a discrete set of values is that the sampling rate cannot take on any arbitrary value. However, this is important for techniques like rate-limiting sampling \cite{Shkuro2019}. Therefore, a workaround is presented in the following, that allows freely choosing the effective sampling rate while still having the span sampling rates from the given discrete set.

\subsection{Rate-Limiting Sampling}
A way to limit the rate of sampled spans is to choose the sampling rate $\symDesiredSampleRate$ proportional to the average time $\symDeltaTime$ between subsequent spans \cite{Shkuro2019}, $\symDesiredSampleRate = \min(1, \symDeltaTime\cdot\symRateLimit )$ where $\symRateLimit$ denotes the rate limit. The average time between subsequent spans needs to be estimated from the timestamps of the most recent spans. In the simplest case, it is just the elapsed time since the last span.

The restriction to a discrete set of sampling rates makes rate-limiting sampling more difficult as the sampling rate cannot be freely chosen. Assume that the desired sampling rate is $\symDesiredSampleRate$ and lies between the two possible sampling rates $\symBase^{\symIndexI+1}$ and $\symBase^{\symIndexI}$ with $\symIndexI\geq 0$ such that $\symDesiredSampleRate\in(\symBase^{\symIndexI+1},\symBase^{\symIndexI}]$. 
The idea is to randomly select one of these two sampling rates for each span. This obviously then leads to an effective sampling rate that is in between.
In particular, if $\symBase^\symIndexI$ is chosen with a probability of $\frac{\symDesiredSampleRate-\symBase^{\symIndexI+1}}{
\symBase^\symIndexI-\symBase^{\symIndexI+1}}$ and $\symBase^{\symIndexI+1}$ with the corresponding complementary probability of $(1-\frac{\symDesiredSampleRate-\symBase^{\symIndexI+1}}{
\symBase^\symIndexI-\symBase^{\symIndexI+1}})$, the effective sampling rate will be 
$\symBase^\symIndexI
\frac{\symDesiredSampleRate-\symBase^{\symIndexI+1}}{
\symBase^\symIndexI-\symBase^{\symIndexI+1}}
+
\symBase^{\symIndexI+1}
(1-\frac{\symDesiredSampleRate-\symBase^{\symIndexI+1}}{
\symBase^\symIndexI-\symBase^{\symIndexI+1}})
=
\symDesiredSampleRate
$ as desired.

\subsection{Span Context}

Partially sampled traces often break into multiple fragments of spans. Therefore, it is feasible to propagate additional information along the trace identifier to allow better reconstruction of hierarchical relationships between spans, even when they cannot be directly lined up.
If the sampling decision is negative, it makes little sense to propagate the span identifier to child spans, as it would be a reference to a non-sampled and therefore later unknown span. Instead, if a span is not sampled, it is more useful to add the identifier of the last sampled ancestor span to the propagated context. In this way, child spans can be linked to the nearest sampled ancestor span. 
When the number of consecutive non-sampled spans is also counted and propagated, even the degree of the ancestor would be known. 

This additional information about hierarchical relationships between spans are very helpful for queries such as counting the number of times a service $\symServiceA$ called some other service $\symServiceB$. If we can determine the hierarchical relationship between $\symServiceA$ and $\symServiceB$ without knowing all the spans in between, the indicator function \eqref{equ:trace_count_extrapolation} will be nonzero for smaller sampled sets of spans, eventually reducing the variance according to \eqref{equ:variance_new}.

\section{Conclusion}
We have presented a theoretical foundation for partial trace sampling that allows sampling rates for spans to be chosen independently without requiring a global service to orchestrate sampling rates at endpoints, as is the case with head-based sampling. 
Given the same local constraints, partial trace sampling can collect more detailed information, especially from rarely called branches. The proposed unbiased estimator can exploit this additional information without knowing the completeness of traces and provably reduces the estimation error in many cases.

\bibliographystyle{ACM-Reference-Format}


\appendix
\section*{Proofs}
\begin{lemma}
\label{the:unbiased}
The estimator given by \eqref{equ:new_estimator} is unbiased, hence $\symExpectation\nolimits_\symRandom(\symNewQuantityEstimator(\symSetOfSampledSpans)) = \symQuantity(\symSetOfSpans)$.
\end{lemma}
\begin{proof}
In the case that $\symRandom\in[\symSampleRate_{\symIndexI},\symSampleRate_{\symIndexI+1})$
we have
$\symQuantity(\symDownSamplingFunction(\symSetOfSampledSpans;\symSampleRateInSample_{\symIndexJ},\symSampleRateFunction))
=
\symQuantity(\symDownSamplingFunction(\symSetOfSpans;\symSampleRateInSample_{\symIndexJ},\symSampleRateFunction))
=
\symQuantity(\symDownSamplingFunction(\symSetOfSpans;\symSampleRate_{\symIndexI+\symIndexJ},\symSampleRateFunction))
=
\symQuantity_{\symIndexI+\symIndexJ}
$ 
for $0\leq\symIndexJ<\symNumSampleRatesInSample=\symNumSampleRates-\symIndexI$. Here we used the relationship  $\symSampleRateInSample_\symIndexJ = \symSampleRate_{\symIndexJ + \symIndexI}$ and introduced the shorthand notation
$\symQuantity_{\symIndexI} := \symQuantity(\symDownSamplingFunction(\symSetOfSpans;\symSampleRate_{\symIndexI},\symSampleRateFunction))$ for $0\leq \symIndexI < \symNumSampleRates$. Together with the definition $\symQuantity_\symNumSampleRates:=0$ the estimate according to \eqref{equ:new_estimator} conditioned on $\symRandom\in[\symSampleRate_{\symIndexI},\symSampleRate_{\symIndexI+1})$ can be written as
\begin{equation*}
\symNewQuantityEstimator(\symSetOfSampledSpans)
\mid
\symRandom\in[\symSampleRate_{\symIndexI},\symSampleRate_{\symIndexI+1})
=
\sum_{\symIndexJ=1}^{\symNumSampleRatesInSample}
\frac{\symQuantity_{\symIndexI+\symIndexJ-1}-\symQuantity_{\symIndexI+\symIndexJ}}{\symSampleRate_{\symIndexI+\symIndexJ}}
=
\sum_{\symIndexJ=\symIndexI+1}^{\symNumSampleRates}
\frac{\symQuantity_{\symIndexJ-1}-\symQuantity_{\symIndexJ}}{\symSampleRate_{\symIndexJ}}.
\end{equation*}
Calculating the expectation by averaging over all possible values of $\symRandom\in[0,1)$, where only values from $[0, \symSampleRate_\symNumSampleRates)$ lead to nonempty sets of sampled spans and therefore can contribute, yields
\begin{equation*}
\begin{aligned}
& \symExpectation\nolimits_\symRandom(\symNewQuantityEstimator(\symSetOfSampledSpans))
=
\\
& =
\sum_{\symIndexI=0}^{\symNumSampleRates-1}
\symProbability(\symRandom \in [\symSampleRate_{\symIndexI}, \symSampleRate_{\symIndexI+1}))
\symExpectation\nolimits_\symRandom(
\symNewQuantityEstimator(\symSetOfSampledSpans)
\mid
\symRandom\in[\symSampleRate_{\symIndexI},\symSampleRate_{\symIndexI+1}))
\\
& =
\sum_{\symIndexI=0}^{\symNumSampleRates-1}
\left(\symSampleRate_{\symIndexI+1}-\symSampleRate_{\symIndexI}\right)
\sum_{\symIndexJ=\symIndexI+1}^{\symNumSampleRates}
\frac{\symQuantity_{\symIndexJ-1}-\symQuantity_{\symIndexJ}}{\symSampleRate_{\symIndexJ}}
\\
& =
\sum_{\symIndexI=0}^{\symNumSampleRates-1}
\sum_{\symIndexJ=\symIndexI+1}^{\symNumSampleRates}
\left(\symSampleRate_{\symIndexI+1}-\symSampleRate_{\symIndexI}\right)
\frac{\symQuantity_{\symIndexJ-1}-\symQuantity_{\symIndexJ}}{\symSampleRate_{\symIndexJ}}
\\
& =
\sum_{\symIndexJ=1}^{\symNumSampleRates}
\sum_{\symIndexI=0}^{\symIndexJ-1}
\left(\symSampleRate_{\symIndexI+1}-\symSampleRate_{\symIndexI}\right)
\frac{\symQuantity_{\symIndexJ-1}-\symQuantity_{\symIndexJ}}{\symSampleRate_{\symIndexJ}}
=
\sum_{\symIndexJ=1}^{\symNumSampleRates}
\left(\symSampleRate_{\symIndexJ}-\symSampleRate_{0}\right)
\frac{\symQuantity_{\symIndexJ-1}-\symQuantity_{\symIndexJ}}{\symSampleRate_{\symIndexJ}}
\\
& =
\sum_{\symIndexJ=1}^{\symNumSampleRates}
\symQuantity_{\symIndexJ-1}-\symQuantity_{\symIndexJ}
=
\symQuantity_{0}-\symQuantity_{\symNumSampleRates}
=
\symQuantity_{0}
=
\symQuantity(\symDownSamplingFunction(\symSetOfSpans;0,\symSampleRateFunction))
=
\symQuantity(\symSetOfSpans).
\end{aligned}
\end{equation*}
\end{proof}

\begin{lemma}
\label{the:variance}
The variance of estimator \eqref{equ:new_estimator} is given by 
\begin{align*}
&\symVariance\nolimits_\symRandom(\symNewQuantityEstimator(\symSetOfSampledSpans))=
\\
&=\symQuantity(\symSetOfSpans)^2
\left(
\frac{1}{\symSampleRate_{\symNumSampleRates}}
-
1
\right)
+
\sum_{\symIndexJ=1}^{\symNumSampleRates-1}
(\symQuantity(\symSetOfSpans)-\symQuantity(\symDownSamplingFunction(\symSetOfSpans;\symSampleRate_{\symIndexJ},\symSampleRateFunction)))^2
\left(
\frac{1}{\symSampleRate_{\symIndexJ}}
-
\frac{1}{\symSampleRate_{\symIndexJ+1}}
\right)
\end{align*}
\end{lemma}

\begin{proof}
Using the same definitions as in \cref{the:unbiased}, the second moment is calculated as
\begin{align*}
&\symExpectation\nolimits_\symRandom((\symNewQuantityEstimator(\symSetOfSampledSpans))^2)
=
\\
&=
\sum_{\symIndexI=0}^{\symNumSampleRates-1}
\left(\symSampleRate_{\symIndexI+1}-\symSampleRate_{\symIndexI}\right)
\symExpectation\nolimits_\symRandom(
(\symNewQuantityEstimator(\symSetOfSampledSpans))^2
\mid
\symRandom\in[\symSampleRate_{\symIndexI},\symSampleRate_{\symIndexI+1}))
\\
&=
\sum_{\symIndexI=0}^{\symNumSampleRates-1} 
(\symSampleRate_{\symIndexI+1}-\symSampleRate_{\symIndexI}) 
\left(
\sum_{\symIndexJ=\symIndexI+1}^{\symNumSampleRates}
\frac{\symQuantity_{\symIndexJ-1}
-
\symQuantity_{\symIndexJ}
}{\symSampleRate_{\symIndexJ}}
\right)^2
\\
&=
\sum_{\symIndexI=0}^{\symNumSampleRates-1} 
\sum_{\symIndexJ=\symIndexI+1}^{\symNumSampleRates}
(\symSampleRate_{\symIndexI+1}-\symSampleRate_{\symIndexI}) 
\left(
\frac{\symQuantity_{\symIndexJ-1}
-
\symQuantity_{\symIndexJ}
}{\symSampleRate_{\symIndexJ}}
\right)^2
\\
&\qquad+
2
\sum_{\symIndexI=0}^{\symNumSampleRates-1} 
\sum_{\symIndexJ=\symIndexI+1}^{\symNumSampleRates-1}
\sum_{\symIndexK=\symIndexJ+1}^{\symNumSampleRates}
(\symSampleRate_{\symIndexI+1}-\symSampleRate_{\symIndexI}) 
\left(
\frac{\symQuantity_{\symIndexJ-1}
-
\symQuantity_{\symIndexJ}
}{\symSampleRate_{\symIndexJ}}
\right)
\left(
\frac{\symQuantity_{\symIndexK-1}
-
\symQuantity_{\symIndexK}
}{\symSampleRate_{\symIndexK}}
\right)
\\
&=
\sum_{\symIndexJ=1}^{\symNumSampleRates}
\sum_{\symIndexI=0}^{\symIndexJ-1} 
(\symSampleRate_{\symIndexI+1}-\symSampleRate_{\symIndexI}) 
\left(
\frac{\symQuantity_{\symIndexJ-1}
-
\symQuantity_{\symIndexJ}
}{\symSampleRate_{\symIndexJ}}
\right)^2
\\
&\qquad +
2
\sum_{\symIndexJ=1}^{\symNumSampleRates-1}
\sum_{\symIndexK=\symIndexJ+1}^{\symNumSampleRates}
\sum_{\symIndexI=0}^{\symIndexJ-1} 
(\symSampleRate_{\symIndexI+1}-\symSampleRate_{\symIndexI}) 
\left(
\frac{\symQuantity_{\symIndexJ-1}
-
\symQuantity_{\symIndexJ}
}{\symSampleRate_{\symIndexJ}}
\right)
\left(
\frac{\symQuantity_{\symIndexK-1}
-
\symQuantity_{\symIndexK}
}{\symSampleRate_{\symIndexK}}
\right)
\\
&=
\sum_{\symIndexJ=1}^{\symNumSampleRates}
\frac{
\left(
\symQuantity_{\symIndexJ-1}
-
\symQuantity_{\symIndexJ}
\right)^2
}{\symSampleRate_{\symIndexJ}}
+
2
\sum_{\symIndexJ=1}^{\symNumSampleRates-1}
\sum_{\symIndexK=\symIndexJ+1}^{\symNumSampleRates}
\frac{(\symQuantity_{\symIndexJ-1}
-
\symQuantity_{\symIndexJ}
)(\symQuantity_{\symIndexK-1}
-
\symQuantity_{\symIndexK})}{\symSampleRate_{\symIndexK}}
\\
&=
\sum_{\symIndexJ=1}^{\symNumSampleRates}
\frac{
\left(
\symQuantity_{\symIndexJ-1}
-
\symQuantity_{\symIndexJ}
\right)^2
}{\symSampleRate_{\symIndexJ}}
+
2
\sum_{\symIndexK=2}^{\symNumSampleRates}
\sum_{\symIndexJ=1}^{\symIndexK-1}
\frac{(\symQuantity_{\symIndexJ-1}
-
\symQuantity_{\symIndexJ}
)(\symQuantity_{\symIndexK-1}
-
\symQuantity_{\symIndexK})}{\symSampleRate_{\symIndexK}}
\\
&=
\sum_{\symIndexJ=1}^{\symNumSampleRates}
\frac{
\left(
\symQuantity_{\symIndexJ-1}
-
\symQuantity_{\symIndexJ}
\right)^2
}{\symSampleRate_{\symIndexJ}}
+
2
\sum_{\symIndexK=2}^{\symNumSampleRates}
\frac{(\symQuantity_{0}
-
\symQuantity_{\symIndexK-1}
)(\symQuantity_{\symIndexK-1}
-
\symQuantity_{\symIndexK})}{\symSampleRate_{\symIndexK}}
\\
&=
\sum_{\symIndexJ=1}^{\symNumSampleRates}
\frac{
\left(
\symQuantity_{\symIndexJ-1}
-
\symQuantity_{\symIndexJ}
\right)^2
}{\symSampleRate_{\symIndexJ}}
+
2
\frac{(\symQuantity_{0}
-
\symQuantity_{\symIndexJ-1}
)(\symQuantity_{\symIndexJ-1}
-
\symQuantity_{\symIndexJ})}{\symSampleRate_{\symIndexJ}}
\\
&=
\sum_{\symIndexJ=1}^{\symNumSampleRates}
\frac{
\symQuantity_{\symIndexJ-1}
-
\symQuantity_{\symIndexJ}
}{\symSampleRate_{\symIndexJ}}
\left(
2
\symQuantity_{0}
-
\symQuantity_{\symIndexJ-1}
-
\symQuantity_{\symIndexJ}
\right)
\\
&=
\sum_{\symIndexJ=1}^{\symNumSampleRates}
\frac{
(\symQuantity_{0}-\symQuantity_{\symIndexJ})^2
-
(\symQuantity_{0}-\symQuantity_{\symIndexJ-1})^2
}
{
\symSampleRate_{\symIndexJ}
}
\\
&=
\frac{(\symQuantity_{0}-\symQuantity_{\symNumSampleRates})^2}{\symSampleRate_{\symNumSampleRates}}
+
\left(
\sum_{\symIndexJ=1}^{\symNumSampleRates-1}
(\symQuantity_{0}-\symQuantity_{\symIndexJ})^2
\left(
\frac{1}{\symSampleRate_{\symIndexJ}}
-
\frac{1}{\symSampleRate_{\symIndexJ+1}}
\right)
\right)
-
\frac{(\symQuantity_{0}-\symQuantity_{0})^2}{\symSampleRate_{1}}
\\
&=
\frac{\symQuantity_{0}^2}{\symSampleRate_{\symNumSampleRates}}
+
\sum_{\symIndexJ=1}^{\symNumSampleRates-1}
(\symQuantity_{0}-\symQuantity_{\symIndexJ})^2
\left(
\frac{1}{\symSampleRate_{\symIndexJ}}
-
\frac{1}{\symSampleRate_{\symIndexJ+1}}
\right).
\end{align*}
Therefore, the variance can be calculated as
\begin{align*}
\symVariance\nolimits_\symRandom(\symNewQuantityEstimator(\symSetOfSampledSpans))
&= 
\symExpectation\nolimits_\symRandom((\symNewQuantityEstimator(\symSetOfSampledSpans))^2)
-
(\symExpectation\nolimits_\symRandom(\symNewQuantityEstimator(\symSetOfSampledSpans)))^2
\\
&=
\frac{\symQuantity_{0}^2}{\symSampleRate_{\symNumSampleRates}}
+
\left(
\sum_{\symIndexJ=1}^{\symNumSampleRates-1}
(\symQuantity_{0}-\symQuantity_{\symIndexJ})^2
\left(
\frac{1}{\symSampleRate_{\symIndexJ}}
-
\frac{1}{\symSampleRate_{\symIndexJ+1}}
\right)
\right)
-
\symQuantity_{0}^2
\\
&=
\symQuantity_{0}^2
\left(
\frac{1}{\symSampleRate_{\symNumSampleRates}}
-
1
\right)
+
\sum_{\symIndexJ=1}^{\symNumSampleRates-1}
(\symQuantity_{0}-\symQuantity_{\symIndexJ})^2
\left(
\frac{1}{\symSampleRate_{\symIndexJ}}
-
\frac{1}{\symSampleRate_{\symIndexJ+1}}
\right).
\end{align*}
\end{proof}

\begin{lemma}
\label{lem:comparison}
If $\symQuantity$ is bounded as defined by \eqref{equ:bounded}, hence  $\symQuantity(\symSetOfSampledSpans)\in[0, 2\symQuantity(\symSetOfSpans)]$ with $\symSetOfSampledSpans = \symDownSamplingFunction(\symSetOfSpans;\symRandom,\symSampleRateFunction)$ for all $\symRandom\in[0,1)$, we will have $\symVariance_\symRandom (\symNewQuantityEstimator(\symSetOfSampledSpans)) \leq \symVariance_\symRandom (\symNaiveQuantityEstimator(\symSetOfSampledSpans))$.
\end{lemma}
\begin{proof}
Since $\symQuantity(\symSetOfSampledSpans)\in[0, 2\symQuantity(\symSetOfSpans)]$ implies $|\symQuantity(\symSetOfSampledSpans) - \symQuantity(\symSetOfSpans)|\leq |\symQuantity(\symSetOfSpans)|$, we have
\begin{align*}
&\symVariance\nolimits_\symRandom(\symNewQuantityEstimator(\symSetOfSampledSpans))=
\\
&=
\symQuantity(\symSetOfSpans)^2
\left(
\frac{1}{\symSampleRate_{\symNumSampleRates}}
-
1
\right)
+
\sum_{\symIndexJ=1}^{\symNumSampleRates-1}
(\symQuantity(\symSetOfSpans)-\symQuantity(\symDownSamplingFunction(\symSetOfSpans;\symSampleRate_{\symIndexJ},\symSampleRateFunction)))^2
\left(
\frac{1}{\symSampleRate_{\symIndexJ}}
-
\frac{1}{\symSampleRate_{\symIndexJ+1}}
\right)
\\
&
\leq
\symQuantity(\symSetOfSpans)^2
\left(
\frac{1}{\symSampleRate_{\symNumSampleRates}}
-
1
\right)
+
\sum_{\symIndexJ=1}^{\symNumSampleRates-1}
\symQuantity(\symSetOfSpans)^2
\left(
\frac{1}{\symSampleRate_{\symIndexJ}}
-
\frac{1}{\symSampleRate_{\symIndexJ+1}}
\right)
\\
&=
\symQuantity(\symSetOfSpans)^2
\left(
\frac{1}{\symSampleRate_1}
-
1
\right)
=
\symVariance\nolimits_\symRandom(\symNaiveQuantityEstimator(\symSetOfSampledSpans)).
\end{align*}
\end{proof}

\balance

\begin{lemma}
\label{the:variance_comparison}
If $\symQuantity$ is monotonic as defined by \eqref{equ:monotonic} and $\symSampleRateFunction_1(\symSpan)\leq \symSampleRateFunction_2(\symSpan)$ holds for all $\symSpan\in\symSetOfSpans$, the variance of estimator \eqref{equ:new_estimator} will satisfy $\symVariance\nolimits_\symRandom(\symNewQuantityEstimator(\symDownSamplingFunction(\symSetOfSpans; \symRandom, \symSampleRateFunction_1))) \geq \symVariance\nolimits_\symRandom(\symNewQuantityEstimator(\symDownSamplingFunction(\symSetOfSpans; \symRandom, \symSampleRateFunction_2)))$.
\end{lemma}
\begin{proof}
Without breaking the monotonicity, we can define $\symQuantity(\emptyset):=0$ to write the variance given in \cref{the:variance} as
\begin{equation*}
\symVariance\nolimits_\symRandom(\symNewQuantityEstimator(\symSetOfSampledSpans))
=
\int_0^1 \frac{
\left( 
\symQuantity(\symSetOfSpans)-\symQuantity(\symDownSamplingFunction(\symSetOfSpans; \symZ, \symSampleRateFunction))
\right)^2
}{
\symZ^2
}
d\symZ.
\end{equation*}
This can be easily verified, when using that $\symDownSamplingFunction(\symSetOfSpans; \symZ, \symSampleRateFunction)$ is piecewise constant within each interval $\symZ\in[\symSampleRate_{\symIndexJ},\symSampleRate_{\symIndexJ+1})$.
$\symSampleRateFunction_1(\symSpan)\leq \symSampleRateFunction_2(\symSpan)$ implies 
$\symDownSamplingFunction(\symSetOfSpans; \symZ, \symSampleRateFunction_1)\subseteq\symDownSamplingFunction(\symSetOfSpans; \symZ, \symSampleRateFunction_2)\subseteq\symSetOfSpans$ for all $\symZ\in[0,1)$. Due to the monotonicity we either have
$\symQuantity(\symDownSamplingFunction(\symSetOfSpans; \symZ, \symSampleRateFunction_1))\leq \symQuantity(\symDownSamplingFunction(\symSetOfSpans; \symZ, \symSampleRateFunction_2))\leq\symQuantity(\symSetOfSpans)$ or $\symQuantity(\symDownSamplingFunction(\symSetOfSpans; \symZ, \symSampleRateFunction_1))\geq \symQuantity(\symDownSamplingFunction(\symSetOfSpans; \symZ, \symSampleRateFunction_2))\geq\symQuantity(\symSetOfSpans)$
and therefore 
$|\symQuantity(\symSetOfSpans)-\symQuantity(\symDownSamplingFunction(\symSetOfSpans; \symZ, \symSampleRateFunction_1))|
\geq
|\symQuantity(\symSetOfSpans)-\symQuantity(\symDownSamplingFunction(\symSetOfSpans; \symZ, \symSampleRateFunction_2))|
\geq
0
$. Therefore, we immediately see that
\begin{align*}
&\symVariance\nolimits_\symRandom(\symNewQuantityEstimator(\symDownSamplingFunction(\symSetOfSpans; \symRandom, \symSampleRateFunction_1)))
=
\int_0^1 \frac{
\left( 
\symQuantity(\symSetOfSpans)-\symQuantity(\symDownSamplingFunction(\symSetOfSpans; \symZ, \symSampleRateFunction_1))
\right)^2
}{
\symZ^2
}
d\symZ
\\
&\geq
\int_0^1 \frac{
\left( 
\symQuantity(\symSetOfSpans)-\symQuantity(\symDownSamplingFunction(\symSetOfSpans; \symZ, \symSampleRateFunction_2))
\right)^2
}{
\symZ^2
}
d\symZ
=
\symVariance\nolimits_\symRandom(\symNewQuantityEstimator(\symDownSamplingFunction(\symSetOfSpans; \symRandom, \symSampleRateFunction_2))).
\end{align*}
\end{proof}

\begin{lemma}
\label{cor:span_extrapolation}
When estimating the number of matching spans from a sampled trace, $\symQuantity(\symSetOfSpans) = |\symSetOfSpans \cap \symOtherSet|$, where $\symOtherSet$ denotes the universe of matching spans, 
the estimator \eqref{equ:new_estimator} reduces to
\begin{equation*}
\symNewQuantityEstimator(\symSetOfSampledSpans)=
\sum_{\symSpan\in\symSetOfSampledSpans\cap\symOtherSet}
\frac{1}{\symSampleRateFunction(\symSpan)}.
\end{equation*}
\end{lemma}

\begin{proof}
Using $\symQuantity(\symSetOfSpans) = |\symSetOfSpans \cap \symOtherSet|$ which implies $\symQuantity(\symDownSamplingFunction(\symSetOfSampledSpans;\symSampleRateInSample_{\symNumSampleRatesInSample},\symSampleRateFunction))=\symQuantity(\emptyset)=0$, we can write \eqref{equ:new_estimator} as
\begin{align*}
\symNewQuantityEstimator(\symSetOfSampledSpans)
&=
\sum_{\symIndexJ=1}^{\symNumSampleRatesInSample}
\frac{\symQuantity(\symDownSamplingFunction(\symSetOfSampledSpans;\symSampleRateInSample_{\symIndexJ-1},\symSampleRateFunction))
-
\symQuantity(\symDownSamplingFunction(\symSetOfSampledSpans;\symSampleRateInSample_{\symIndexJ},\symSampleRateFunction))
}{\symSampleRateInSample_{\symIndexJ}}
\\
&=
\sum_{\symIndexJ=1}^{\symNumSampleRatesInSample}
\frac{
|\symDownSamplingFunction(\symSetOfSampledSpans;\symSampleRateInSample_{\symIndexJ-1},\symSampleRateFunction)\cap\symOtherSet|
-
|\symDownSamplingFunction(\symSetOfSampledSpans;\symSampleRateInSample_{\symIndexJ},\symSampleRateFunction)\cap\symOtherSet|
}{\symSampleRateInSample_{\symIndexJ}}
\\
&=
\sum_{\symIndexJ=1}^{\symNumSampleRatesInSample}
\frac{
|\symDownSamplingFunction(\symSetOfSampledSpans\cap\symOtherSet;\symSampleRateInSample_{\symIndexJ-1},\symSampleRateFunction)\setminus
\symDownSamplingFunction(\symSetOfSampledSpans\cap\symOtherSet;\symSampleRateInSample_{\symIndexJ},\symSampleRateFunction)
|
}{\symSampleRateInSample_{\symIndexJ}}
\\
&=
\sum_{\symIndexJ=1}^{\symNumSampleRatesInSample}
\frac{
|\lbrace\symSpan\in\symSetOfSampledSpans\cap\symOtherSet:\symSampleRateInSample_{\symIndexJ-1} < \symSampleRateFunction(\symSpan)\leq\symSampleRateInSample_{\symIndexJ}\rbrace|
}{\symSampleRateInSample_{\symIndexJ}}
\\
&=
\sum_{\symIndexJ=1}^{\symNumSampleRatesInSample}
\frac{
|\lbrace\symSpan\in\symSetOfSampledSpans\cap\symOtherSet:\symSampleRateFunction(\symSpan)=\symSampleRateInSample_{\symIndexJ}\rbrace|
}{\symSampleRateInSample_{\symIndexJ}}
=
\sum_{\symSpan\in\symSetOfSampledSpans\cap\symOtherSet}
\frac{1}{\symSampleRateFunction(\symSpan)}.
\end{align*}
\end{proof}

\begin{lemma}
\label{cor:indicator_estimation}
If $\symQuantity$ is a monotonically increasing indicator function, which means that $\symQuantity\in\lbrace 0,1\rbrace$, estimator \eqref{equ:new_estimator} simplifies to 
\begin{equation*}
\symNewQuantityEstimator(\symSetOfSampledSpans)
=
\begin{cases}
\frac{1}{\symSampleRateInSample_\symIndexK} & \text{if}\ \symQuantity(\symSetOfSampledSpans) = 1\\
0 & \text{if}\ \symQuantity(\symSetOfSampledSpans) = 0
\end{cases}
\end{equation*}
with  $\symIndexK:=\max\lbrace \symIndexJ\in\lbrace 1,2,\ldots,\symNumSampleRatesInSample\rbrace : \symQuantity(\symDownSamplingFunction(\symSetOfSampledSpans;\symSampleRateInSample_{\symIndexJ-1},\symSampleRateFunction))=1 \rbrace$ and $\symSampleRateInSample_0:=0$.
\end{lemma}

\begin{proof}
The case $\symQuantity(\symSetOfSampledSpans) = 0$ is trivial. 
If $\symQuantity(\symSetOfSampledSpans) = \symQuantity(\symDownSamplingFunction(\symSetOfSampledSpans;\symSampleRateInSample_{0},\symSampleRateFunction)) = 1$, we will have 
\begin{equation*}
\symQuantity(\symDownSamplingFunction(\symSetOfSampledSpans;\symSampleRateInSample_{\symIndexJ},\symSampleRateFunction))
=
\begin{cases}
1 & \text{if}\  \symIndexJ<\symIndexK\\
0 & \text{if}\  \symIndexJ\geq\symIndexK\\
\end{cases} 
\end{equation*}
due to the monotonicity. Using that in \eqref{equ:new_estimator} leads to the claimed identity.
\end{proof}

\end{document}